\documentclass[a4paper,USenglish]{scrartcl}
\def\doi#1{\href{https://doi.org/\detokenize{#1}}{\url{https://doi.org/\detokenize{#1}}}}

\usepackage{amssymb,amsmath,mathrsfs}
\usepackage{amsthm}
\usepackage[algo2e, vlined, ruled, linesnumbered]{algorithm2e}
\usepackage{tikz}
\usepackage{url}

\usepackage[hidelinks]{hyperref}
\usepackage[capitalize,nameinlink]{cleveref}

\newtheorem{theorem}{Theorem}
\newtheorem{lemma}[theorem]{Lemma}
\newtheorem{corollary}[theorem]{Corollary}
\newtheorem{proposition}[theorem]{Proposition}

\newtheorem{observation}[theorem]{Observation}

\newtheorem{definition}[theorem]{Definition}
\newtheorem{problem1}[theorem]{Problem}

\newcommand{\N}{\ensuremath{\mathbb{N}}}
\newcommand{\R}{\ensuremath{\mathcal{R}}}
\newcommand{\Q}{\ensuremath{\mathcal{Q}}}

\newcommand{\NP}{\ensuremath{\mathcal{NP}}}

\newcommand{\A}{\ensuremath{\mathcal{A}}}

\renewcommand{\O}{\ensuremath{{\cal O}}}
\newcommand{\cf}{\ensuremath{{\cal F}}}
\newcommand{\cl}{\ensuremath{{\cal L}}}

\newcommand\mdoubleplus{\ensuremath{\mathbin{+\mkern-3mu+}}}

\usepackage{enumitem}

\begin{document}

\title{Linearizing Partial Search Orders}
	
\author{Robert Scheffler}
	
\date{Institute of Mathematics, Brandenburg University of Technology, Cottbus, Germany \\ \texttt{robert.scheffler@b-tu.de}}
	
  	\maketitle

 	\begin{abstract}
		In recent years, questions about the construction of special orderings of a given graph search were studied by several authors. On the one hand, the so called end-vertex problem introduced by Corneil et al. in 2010 asks for search orderings ending in a special vertex. On the other hand, the problem of finding orderings that induce a given search tree was introduced already in the 1980s by Hagerup and received new attention most recently by Beisegel et al. Here, we introduce a generalization of some of these problems by studying the question whether there is a search ordering that is a linear extension of a given partial order on a graph's vertex set. We show that this problem can be solved in polynomial time on chordal bipartite graphs for LBFS, which also implies the first polynomial-time algorithms for the end-vertex problem and two search tree problems for this combination of graph class and search. Furthermore, we present polynomial-time algorithms for LBFS and MCS on split graphs which generalize known results for the end-vertex and search tree problems.
	\end{abstract} 	
	
	\section{Introduction}

The graph searches \emph{Breadth First Search} (BFS) and \emph{Depth First Search} (DFS) are considered as some of the most basic algorithms in both graph theory and computer science. Taught in many undergraduate courses around the world, they are an elementary component of several graph algorithms. There are also many other more sophisticated graph searches, e.g., the \emph{Lexicographic Breadth First Search} (LBFS)~\cite{rose1976algorithmic} or the \emph{Maximum Cardinality Search} (MCS)~\cite{tarjan1984simple} which are also used to solve several graph problems among them the recognition problems of graph classes~\cite{bretscher2008simple,chu2008simple,corneil2009lbfs,dusart2017new}, the computation of minimal separators~\cite{kumar1998minimal} as well as the computation of minimal triangulations~\cite{berry2004maximum}.

In recent years, different problems of finding special search orderings have gained attention from several researchers. One of these problems is the end-vertex problem introduced in 2010 by Corneil et~al.~\cite{corneil2010end}. It asks whether a given vertex in a graph can be visited last by some graph search. The problem was motivated by \emph{multi-sweep algorithms} where a search is applied several times to a graph such that every application starts in the vertex where the preceding search ordering has ended. Corneil et al.~\cite{corneil2010end} showed that the end-vertex problem is \NP-complete for LBFS on weakly chordal graphs. Similar results were obtained for other searches such as BFS, DFS and MCS~\cite{beisegel2019end-vertex,cao2019graph,charbit2014influence}, while for several graph classes, among them split graphs, polynomial-time algorithms were presented in~\cite{beisegel2019end-vertex,cao2019graph,charbit2014influence,corneil2010end,gorzny2017end-vertices}.

An important structure closely related to a graph search is the corresponding search tree. Such a tree contains all the vertices of the graph and for every vertex different from the start vertex exactly one edge to a vertex preceding it in the search ordering. Such trees can be of particular interest as for instance the tree obtained by a BFS contains the shortest paths from the root to all other vertices in the graph and the trees generated by DFS can be used for fast planarity testing~\cite{hopcroft1974efficient}. The problem of deciding whether a given spanning tree of a graph can be obtained by a particular search was introduced by Hagerup~\cite{hagerup1985biconnected} in 1985, who presented a linear-time algorithm for recognizing DFS-trees. In the same year, Hagerup and Nowak~\cite{hagerup1985recognition} gave a similar result for the BFS-tree recognition. 
Recently, Beisegel et al.~\cite{beisegel2019recognizing,beisegel2021recognition} introduced a more general framework for the search tree recognition problem. They introduced the term \cf-tree for search trees where a vertex is connected to its first visited neighbor, i.e., BFS-like trees, and \cl-trees for search trees where a vertex is connected to its most recently visited neighbor, i.e., DFS-like trees. They showed, among other things, that the \cf-tree recognition is \NP-complete for LBFS and MCS on weakly chordal graphs, while the problem can be solved in polynomial time for both searches on chordal graphs.

\paragraph*{Our Contribution.} 
There seems to be a strong relationship between the complexity of the end-vertex problem and the recognition problem of $\cf$-trees. There are many combinations of graph classes and graph searches where both problems are \NP-complete or both problems are solvable in polynomial time. To further study this relationship, we present a generalization of these two problems by introducing the \emph{Partial Search Order Problem (PSOP)} of a graph search $\A$. Given a graph $G$ and a partial order $\pi$ on its vertex set, it asks whether there is a search ordering produced by $\A$ which is a linear extension of $\pi$. We show that a greedy algorithm solves the PSOP of Generic Search, i.e., the search where every vertex can be visited next as long as it has an already visited neighbor.  Furthermore, we present a polynomial-time algorithm for the PSOP of LBFS on chordal bipartite graphs, i.e., bipartite graphs without induced cycles of length larger than four. This result also implies a polynomial-time algorithm for the end-vertex problem on chordal bipartite graphs, a generalization of the result by Gorzny and Huang~\cite{gorzny2017end-vertices} on the end-vertex problem of LBFS on AT-free bipartite graphs. For split graphs, we will give polynomial-time algorithms for the PSOP of LBFS and MCS that generalize the results on the end-vertex problem~\cite{beisegel2019end-vertex,charbit2014influence} and the \cf-tree problem~\cite{beisegel2019recognizing} of these searches on this graph class.

	\section{Preliminaries}

\paragraph*{General Notation.}
The graphs considered in this paper are finite, undirected, simple and connected. Given a graph $G$, we denote by $V(G)$ the set of vertices and by $E(G)$ the set of edges. For a vertex $v\in V(G)$, we denote by $N(v)$ the \emph{neighborhood} of $v$ in $G$, i.e., the set $N(v)=\{u\in V\mid uv\in E\}$, where an edge between $u$ and $v$ in $G$ is denoted by $uv$. The \emph{neighborhood of a set $A \subset V(G)$} is the union of the neighborhoods of the vertices in $A$. The \emph{distance} of a vertex $v$ to a vertex $w$ is the number of edges of the shortest path from $v$ to $w$. The set $N^\ell(v)$ contains all vertices whose distance to the vertex $v$ is equal to $\ell$.

A \emph{clique} in a graph $G$ is a set of pairwise adjacent vertices and an \emph{independent set} in $G$ is a set of pairwise nonadjacent vertices. A \emph{split graph} $G$ is a graph whose vertex set can be partitioned into sets $C$ and $I$, such that $C$ is a clique in $G$ and $I$ is an independent set in $G$. We call such a partition a \emph{split partition}. A graph is \emph{bipartite} if its vertex set can be partitioned into two independent sets $X$ and $Y$. A bipartite graph $G$ is called \emph{chordal bipartite} if every induced cycle contained in $G$ has a length of four. Note that there is a strong relationship between split graphs and bipartite graphs. Every bipartite graph is a spanning subgraph of a split graph and every split graph can be made to a bipartite graph by removing the edges between the clique vertices.

A \emph{tree} is an acyclic connected graph. A \emph{spanning tree} of a graph $G$ is an acyclic connected subgraph of $G$ which contains all vertices of $G$. A tree together with a distinguished \emph{root vertex} $r$ is said to be \emph{rooted}. In such a rooted tree $T$, a vertex $v$ is the \emph{parent} of vertex $w$ if $v$ is an element of the unique path from $w$ to the root $r$ and the edge $vw$ is contained in $T$. A vertex $w$ is called the \emph{child} of $ v $ if $ v $ is the parent of $ w $.

A \emph{vertex ordering} of $G$ is a bijection $\sigma: \{1,2,\dots,|V(G)|\}\rightarrow V(G)$. We denote by $\sigma^{-1}(v)$ the position of vertex $v\in V(G)$. Given two vertices $u$ and $v$ in $G$ we say that $u$ is \emph{to the left} (resp. \emph{to the right}) of $v$ if $\sigma^{-1}(u)<\sigma^{-1}(v)$ (resp. $\sigma^{-1}(u)>\sigma^{-1}(v)$) and we denote this by $u \prec_{\sigma}v$ (resp.  $u \succ_{\sigma}v$).

A \emph{partial order} $ \pi $ on a set $X$ is a reflexive, antisymmetric and transitive binary relation on $ X $. We also denote $(x,y) \in \pi$ by $x \prec_\pi y$ if $x \neq y$. A \emph{linear extension} of $\pi$ is a total order $\sigma$ of the elements of $X$ that fulfills all conditions of $\pi$, i.e., if $x \prec_\pi y$, then $x \prec_\sigma y$. We will often use the term ``$\sigma$ extends $\pi$''. For a binary relation $ \pi' $ on $ X $ we say that the \emph{reflexive and transitive closure} of $ \pi' $ is the smallest binary relation $ \pi \supseteq \pi' $ that is reflexive and transitive.

\paragraph*{Graph Searches.}
A \emph{graph search} is an algorithm that, given a graph $G$ as input, outputs a vertex ordering of $G$. All graph searches considered in this paper can be formalized adapting a framework introduced by Corneil et al.~\cite{corneil2016tie} (a similar framework is given in~\cite{krueger2011general}). This framework uses subsets of $\N^+$ as vertex labels. Whenever a vertex is visited, its index in the search ordering is added to the labels of its unvisited neighbors. The search $\cal A$ is defined via a strict partial order $\prec_{\cal A}$ on the elements of ${\cal P}(\N^+)$ (see \cref{algo:ls}). For a given graph search $\cal A$ we say that a vertex ordering $\sigma$ of a graph $G$ is an \emph{$\cal A$-ordering} of $G$ if $\sigma$ can be the output of $\cal A$ with input $G$.

\newcommand{\searchlabel}{\emph{label}}
\begin{algorithm2e}[h]
\small
    \KwIn{A graph $G$} 
    \KwOut{A search ordering $\sigma$ of $G$}
    \Begin{
		\lForEach{$v \in V(G)$}{\searchlabel($v$) $\leftarrow$ $\emptyset$}
		\For{$i$ $\leftarrow$ $1$ \KwTo $|V(G)|$}{
			\emph{Eligible} $\leftarrow$ $\{x \in V(G)~|~x$ unnumbered and $\nexists$ unnumbered $y \in V(G)$ \\ \mbox{}\phantom{\emph{Eligible} $\leftarrow$ $\{x \in V(G)~|~$}such that \searchlabel($x$) $\prec_{\cal A}$ \searchlabel($y$)$\}$\;
			let $v$ be an arbitrary vertex in \emph{Eligible}\;\label{line:ls}
			$\sigma(i)$ $\leftarrow$ $v$\tcc*{assigns to $v$ the number $i$}
			\lForEach{unnumbered vertex $w \in N(v)$}{\searchlabel($w$) $\leftarrow$ \searchlabel($w$) $\cup$ $\{i\}$}
		}
	}
    \caption{Label Search($\prec_{\cal A}$)}\label{algo:ls}
\end{algorithm2e}

In the following, we define the searches considered in this paper by presenting suitable partial orders $\prec_{\cal A}$ (see~\cite{corneil2016tie}). The \emph{Generic Search} (GS) is equal to the Label Search($\prec_{GS}$) where $A \prec_{GS} B$ if and only if $A = \emptyset$ and $B \neq \emptyset$. Thus, any vertex with a numbered neighbor can be numbered next.

The partial label order $\prec_{BFS}$ for \emph{Breadth First Search} (BFS) is defined as follows: $A \prec_{BFS} B$ if and only if $A = \emptyset$ and $B \neq \emptyset$ or $\min(A) > \min (B)$. For the \emph{Lexicographic Breadth First Search} (LBFS)~\cite{rose1976algorithmic} we consider the partial order $\prec_{LBFS}$ with $A \prec_{LBFS} B$ if and only if $A \subsetneq B$ or $\min(A \setminus B) > \min(B \setminus A)$. 

The \emph{Maximum Cardinality Search} (MCS)~\cite{tarjan1984simple} uses the partial order $\prec_{MCS}$ with $A \prec_{MCS} B$ if and only if $|A| < |B|$. The \emph{Maximal Neighborhood Search} (MNS)~\cite{corneil2008unified} is defined using $\prec_{MNS}$ with $A \prec_{MNS} B$ if and only if $A \subsetneq B$. If $A \prec_{MNS} B$, then it also holds that $A \prec_{LBFS} B$ and $A \prec_{MCS} B$. Thus, any ordering produced by LBFS or MCS is also an MNS ordering.

Search orderings of these searches can be characterized using so-called \emph{4-point conditions} (see~\cite{corneil2008unified}). The condition for LBFS is given in the following lemma. We will use the condition of LBFS given in the following lemma several times throughout the paper.

\begin{lemma}[\cite{corneil2008unified}]\label{lemma:4point-lbfs}
A vertex ordering $\sigma$ of a graph $G$ is an LBFS ordering if and only if the following property holds: For all vertices $a,b,c \in V(G)$ with $a \prec_\sigma b \prec_\sigma c$, $ac \in E(G)$ and $ab \notin E(G)$ there is a vertex $d \in V(G)$ with $d \prec_\sigma a$, $bd \in E(G)$  and $cd \notin E(G)$.
\end{lemma}

In the search algorithms following the framework given in \cref{algo:ls}, any of the vertices in the set \emph{Eligible} can be chosen as the next vertex. Some applications use special variants of these searches that involve tie-breaking. For any instantiation $\cal A$ of \cref{algo:ls}, we define the graph search ${\cal A}^+$ as follows: Add a vertex ordering $\rho$ of graph $G$ as additional input and replace line~\ref{line:ls} in \cref{algo:ls} with ``let $v$ be the vertex in \emph{Eligible} that is leftmost in $\rho$''. Note that this corresponds to the algorithm TBLS given in~\cite{corneil2016tie}. The search ordering ${\cal A}^+(\rho)$ is unique since there are no ties to break.

The following lemma gives a strong property for all $\A^+$-searches considered here and will be used several times.

\begin{lemma}\label{lemma:plus}
Let $G$ be a graph and $\rho$ be a vertex ordering of $G$. Let ${\cal A}$ be a graph search in $\{$GS, BFS, LBFS, MCS, MNS$\}$ and $\sigma$ the ${\cal A}^+(\rho)$ ordering of $G$. If $u \prec_\sigma v$ and $v \prec_\rho u$, then there is a vertex $x$ with $x \prec_\sigma u$ and $xu \in E(G)$ and $xv \notin E(G)$.  
\end{lemma}

\begin{proof}
Assume for contradiction that for any vertex $x$ with $x \prec_\sigma u$ and $xu \in E(G)$ it also holds $xv \in E(G)$. It can be checked easily that the strict partial order $\prec_{\cal A}$ given for any of the given graph searches ${\cal A}$ has the following property: If $A \subseteq B$, then for any $C$ with $B \prec_{\cal A} C$ it also holds $A \prec_{\cal A} C$.  Consider the step $i$ in the computation of $\sigma$ where $u$ was numbered by the search $\cal A$. Since $v \prec_\rho u$, the vertex $v$ was not an element of \emph{Eligible}. Therefore, there was an unnumbered vertex $x$ with $\searchlabel(v) \prec_{\cal A} \searchlabel(x)$. Due to the assumption on $u$ and $v$, it held that $\searchlabel(u) \subseteq \searchlabel(v)$. Due to the property of $\prec_{\cal A}$ mentioned above, it held that $\searchlabel(u) \prec_{\cal A} \searchlabel(x)$ and $u$ was not an element of \emph{Eligible}; a contradiction.
\end{proof}

	\section{The Partial Search Order Problem}

We start this section by introducing the problem considered in this paper.

\begin{problem1}{Partial Search Order Problem (PSOP) of graph search $\A$}
\begin{description}
\item[\textbf{Instance:}] A graph $G$, a partial order $\pi$ on $V(G)$.
\item[\textbf{Task:}]
Decide whether there is an $\cal A$-ordering of $G$ that extends $\pi$.
 \end{description}
\end{problem1}

We will also consider a special variant, where the start vertex of the search ordering is fixed. We call this problem the \emph{rooted partial search order problem}. Note that the general problem and the rooted problem are polynomial time equivalent. If we have a polynomial time algorithm to solve the rooted problem we can apply it $|V(G)|$ times to solve the general problem. On the other hand, the rooted problem with fixed start vertex $r$ can be solved by a general algorithm. To this end, we add all the tuples $(r,v)$, $v \in V(G)$, to the partial order $\pi$. Note that in the following we always assume that a given start vertex $r$ is a minimal element of the partial order $\pi$ since otherwise we can reject the input immediately.

The \emph{end-vertex problem} of a graph search $\cal A$ introduced by Corneil et al.~\cite{corneil2010end} in 2010 asks whether the vertex $t$ can be the last vertex of an $\cal A$-ordering of a given graph $G$. This question can be encoded by the partial order $\pi := \{(u,v)~|~u,v \in V(G), u = v \text{ or } v=t\}$, leading to the following observation.

\begin{observation}\label{obs:end-vertex}
The end-vertex problem of a graph search $\cal A$ on a graph $G$ can be solved by solving the PSOP of $\cal A$ on $G$ for a partial order of size $\O(|V(G)|)$.
\end{observation}

From this observation it follows directly that the partial search order problem is \NP-complete for BFS, DFS, LBFS, LDFS, MCS and MNS~\cite{beisegel2019end-vertex,charbit2014influence,corneil2010end}.

In~\cite{beisegel2019recognizing}, Beisegel et al. introduced the terms \emph{\cf-tree} and \emph{\cl-tree} of a search ordering. For this we only consider search orderings produced by a \emph{connected graph searches}, i.e., a graph search that outputs search orderings of the Generic Search. In the \cf-tree of such an ordering, every vertex different from the start vertex is connected to its leftmost neighbor in the search ordering. In the \cl-tree, any vertex $v$ different from the start vertex is connected to its rightmost neighbor that is to the left of $v$ in the search ordering. The problem of deciding whether a given spanning tree of a graph can be the \cf-tree (\cl-tree) of a search ordering of a given type is called \emph{\cf-tree (\cl-tree) recognition problem}. If the start vertex is fixed, it is called the \emph{rooted \cf-tree (\cl-tree) recognition problem}. The rooted \cf-tree recognition problem is a special case of the (rooted) PSOP, as the following proposition shows. 

\begin{proposition}\label{prop:f-trees}
Let $\cal A$ be a connected graph search. Given a graph $G$ and a spanning tree $T$ of $G$ rooted in $r$, we define $\pi$ to be the reflexive, transitive closure of the relation $\R := \{(x,y)~|~x$ is parent of $y$ in $T$ or there is child $z$ of $x$ in $T$ with $yz \in E(G)\}$. The tree $T$ is the \cf-tree of an $\cal A$-ordering $\sigma$ of $G$ if and only if $\pi$ is a partial order and $\sigma$ extends $\pi$. 

Therefore, the rooted \cf-tree problem of a graph search $\cal A$ on a graph $G$ can be solved by solving the (rooted) PSOP of $\cal A$ on $G$.
\end{proposition}

\begin{proof}
Assume $T$ is the \cf-tree of the $\cal A$-ordering $\sigma$ starting in $r$. Then for any vertex $y \in V(G)$ with parent $x \in V(G)$ it holds that $x \prec_\sigma y$. Furthermore, for any neighbor $z$ of $y$ different from $x$ it holds that $x \prec_\sigma z$ since otherwise the edge $xy$ would not be an element of $T$. Therefore, $\sigma$ must be a linear extension of $\R$ and, thus, of $\pi$. 

On the other hand, let the $\cal A$-ordering $\sigma$ be a linear extension of $\pi$. Let $y$ be an arbitrary vertex in $G$ and let $x$ be the parent of $y$ in $T$. Following from the construction of $\R$ and $\pi$, the vertex $x$ is the leftmost neighbor of $y$ in $\sigma$. Therefore, the edge $xy$ is part of the \cf-tree of $\sigma$. Hence, this \cf-tree must be equal to $T$.
\end{proof}

Note that the general \cf-tree recognition problem without fixed start vertex can be solved by deciding the partial search order problem for any possible root. The \cl-tree recognition problem, however, is not a special case of the partial search order problem. For a vertex $w$, its parent $v$ and another neighbor $z$ of $w$, it must either hold that $v \prec_\sigma w \prec_\sigma z$ or that $z \prec_\sigma v \prec_\sigma w$. These constraints cannot be encoded using a partial order. Nevertheless, we will see in Section~\ref{sec:chordal-bipartite} that on bipartite graphs the PSOP of (L)BFS is a generalization of the \cl-tree recognition problem of (L)BFS.

We conclude this section with a simple algorithm for the rooted PSOP of Generic Search (see \cref{algo:generic} for the pseudocode). First the algorithm visits the given start vertex $r$. Afterwards it looks for a vertex with an already visited neighbor among all vertices that are minimal in the remaining partial order. If no such vertex exists, then it rejects. Otherwise, it visits one of these vertices next.

\begin{algorithm2e}[t]
\small
    \KwIn{Connected graph $G$, a vertex $r \in V(G)$, a partial order $\pi$ on $V(G)$} 
    \KwOut{GS ordering $\sigma$ of $G$ extending $\pi$ or ``$\pi$ cannot be linearized''}
    \Begin{
		$S$ $\leftarrow$ $\{r\}$; \qquad $i \leftarrow 1$\;
		\While{$S \neq \emptyset$}{
			let $v$ be an arbitrary element of $S$\;
			remove $v$ from $S$ and from $\pi$\;
			$\sigma(i)$ $\leftarrow$ $v$; \qquad $i$ $\leftarrow$ $i + 1$\;
			\lForEach{$w \in N(v)$}{mark $w$}
			\lForEach{marked $x \in V(G)$ which is minimal in $\pi$}{
				$S$ $\leftarrow$ $S \cup \{x\}$}
		}
        \lIf{$i = |V(G)|$ + 1}{\Return{$\sigma$}} \lElse{\Return{``$\pi$ cannot be linearized''}}
	}
    \caption{Rooted PSOP of Generic Search}\label{algo:generic}
\end{algorithm2e}

\begin{theorem}\label{thm:algo:generic}
\cref{algo:generic} solves the rooted partial search order problem of Generic Search for a graph $G$ and a partial order $\pi$ in time $\O(|V(G)|+|E(G)| + |\pi|)$.
\end{theorem}

\begin{proof}
In \cref{algo:generic} every vertex $x$ different from $r$ is not inserted into $S$ before a neighbor of $x$ was numbered. As only vertices contained in $S$ are visited, we know that every ordering $\sigma$ returned by  \cref{algo:generic} is a Generic Search ordering of $G$. Furthermore, $\sigma$ is a linear extension of $\pi$ since a vertex $v$ is only added to $\sigma$ if for any $(x,v) \in \pi$ it holds that $x$ is already numbered in $\sigma$.

If \cref{algo:generic} returns ``$\pi$ cannot be linearized'', then there is at least one vertex in $V(G)$ that was not inserted into $S$. Assume for contradiction that there is a Generic Search ordering $\sigma$ of $G$ starting in $r$ that is a linear extension of $\pi$ and let $v$ be the leftmost vertex in $\sigma$ that was not inserted into $S$. Since $\sigma$ is a GS ordering, there is a neighbor $w$ of $v$ which is to the left of $v$ in $\sigma$. Vertex $w$ has been added to $S$ and, therefore, $v$ was marked by \cref{algo:generic}. As $v$ was not inserted into $S$, there must be a vertex $x$ that was also not inserted into $S$ with $x \prec_\pi v$. However, due to the choice of $v$, vertex $x$ must be to the right of $v$ in $\sigma$. This is contradiction since $\sigma$ was chosen to be a linear extension of $\pi$.

It remains to show that \cref{algo:generic} can be implemented such that it has a linear running time. We encode $\pi$ with a directed graph $D_\pi$, i.e., the edge $(x,y) \in E(D_\pi)$ if and only if $x \prec_\pi y$. Whenever we visit a vertex $v$ we iterate through its neighborhood in $G$ and mark the neighbors. If there is a neighbor $w$ that is minimal in $\pi$, i.e., its in-degree in $D_\pi$ is zero, then we add $w$ to $S$. Afterwards we iterate through the out-neighborhood of $v$ in $D_\pi$. If there is a marked out-neighbor with in-degree one, then we add this vertex to $S$. Afterwards, we delete $v$ from $D_\pi$. Summarizing, the algorithm iterates through every neighborhood in $G$ and through every out-neighborhood in $D_\pi$ at most once. Therefore, the total running time is bounded by $\O(|V(G)| + |E(G)| + |\pi|)$. 
\end{proof}

	\section{One-Before-All Orderings}

Before we present algorithms for the PSOP we introduce a new ordering problem that will be used in the following two sections to solve the PSOP of LBFS on both chordal bipartite graphs and split graphs.

\begin{problem1}{One-Before-All Problem (OBAP)}
\begin{description}
\item[\textbf{Instance:}] A set $M$, a set $\mathcal{Q} \subseteq \mathcal{P}(M)$, a relation $\mathcal{R} \subseteq \mathcal{Q} \times \mathcal{Q}$
\item[\textbf{Task:}]
Decide whether there is a linear ordering $\sigma$ of $M$ fulfilling the \emph{One-Before-All property}, i.e., for all $A,B \in \mathcal{Q}$  with $(A,B) \in \mathcal{R}$ and $B \neq \emptyset$ there is an $x \in A$ such that for all $y \in B$ it holds that $x \prec_\sigma y$.
 \end{description}
\end{problem1}

Note that every partial order $\pi$ on a set $X$ can be encoded as an OBAP instance by setting $M = X$, $\Q = \{\{x\}~|~x \in X\}$ and $\R = \{(\{x\},~\{y\})~|~x \prec_\pi y\}$. Thus, the OBAP generalizes the problem of finding a linear extension of an partial order.

\begin{algorithm2e}[h]
\small
    \KwIn{A set $M$, a set $\mathcal{Q} \subseteq \mathcal{P}(M)$, a relation $\mathcal{R} \subseteq \mathcal{Q} \times \mathcal{Q}$.} 
    \KwOut{An OBA-ordering $\sigma$ of the elements in $M$ or ``No ordering".}
    \Begin{    
            $r(A) \leftarrow 0 \quad \forall A \in \mathcal{Q}$; \qquad $t(x) \leftarrow 0 \quad \forall x \in M$; \qquad $S \leftarrow \emptyset$; \qquad $i \leftarrow 1$\;
            
            \lForEach{$(A,B) \in \mathcal{R}$}{$r(B) \leftarrow r(B) + 1$}
            
            \ForEach{$A \in \mathcal{Q}$ with $r(A) > 0$}{
                \lForEach{$x \in A$}{$t(x) \leftarrow t(x) + 1$}
            }
            
            \lForEach{$x \in M$ with $t(x) = 0$}{$S \leftarrow S \cup \{x\}$}
            
            \While{$S \neq \emptyset$}{
                let $x$ be an element in $S$\;
                $S \leftarrow S \setminus \{x\}$; \qquad
                $\sigma(i) \leftarrow x$; \qquad
                $i \leftarrow i + 1$\;
                \ForEach{$A \in \mathcal{Q}$ with $x \in A$}{\label{line:obaop-for1}
                    $\cal Q$ $\leftarrow$ ${\cal Q} \setminus \{A\}$\;
                    \ForEach{$(A,B) \in \mathcal{R}$}{\label{line:obaop-for2}
                        $\mathcal{R} \leftarrow \mathcal{R} \setminus \{(A,B)\}$\;
                        $r(B) \leftarrow r(B) - 1$\;
                        \If{$r(B) = 0$}{
                            \ForEach{$y \in B$} {\label{line:obaop-for3}
                                $t(y) \leftarrow t(y) - 1$\;
                                \lIf{$t(y) = 0$}{$S \leftarrow S \cup \{y\}$}
                            }
                        }
                        
                    }
                }
            }
            
            \lIf{$i = |M| + 1$}{\Return{$\sigma$}}
            \lElse{\Return{``No ordering"}}
            
        }
    \caption{OBAP}\label{algo:obaop}
\end{algorithm2e}

In the following we describe how we can solve the one-before-all problem in time linear in the input size $|M| + |\mathcal{R}| + \sum_{A \in \mathcal{Q}} |A|$ (see \cref{algo:obaop} for the pseudocode). For every set $A \in {\cal Q}$ we introduce a counter $r(A)$ containing the number of tuples $(X,A) \in {\cal R}$. For every element $x \in M$ the variable $t(x)$ counts the number of sets $A \in {\cal Q}$ with $x \in A$ and $r(A) > 0$. Our algorithm builds the ordering $\sigma$ from left to right. It is not difficult to see that an element $x$ can be chosen next if and only if $t(x) = 0$. As long as such an element exists, the algorithm chooses one, deletes all tuples $(A,B)$ with $x \in A$ from $\R$ and updates the $r$- and the $t$-values. If no such element exists, then the algorithm returns ``No ordering".

\begin{theorem}\label{thm:algo-obaop}
Given a set $M$, a set $\mathcal{Q} \subseteq \mathcal{P}(M)$ and a relation $\mathcal{R} \subseteq \mathcal{Q} \times \mathcal{Q}$, \cref{algo:obaop} returns a linear ordering $\sigma$ of $M$ fulfilling the one-before-all property if and only if such an ordering exists. The running time of the algorithm is $\O(|M| + |\mathcal{R}| + \sum_{A \in \mathcal{Q}} |A|)$.
\end{theorem}

\begin{proof}
Assume \cref{algo:obaop} returns the ordering $\sigma$ that is not a feasible OBA ordering. Then there are sets $A,B \in {\cal Q}$ with $(A,B) \in {\cal R}$ such that there is an $x \in B$ that is to the left of any element of $A$ in $\sigma$. Tuples are only deleted from ${\cal R}$ if an element of the left set of the tuple received an index in $\sigma$. Therefore, the tuple $(A,B)$ was still in ${\cal R}$ when $x$ was inserted into $S$. This is a contradiction, because $t(x)$ would have been larger than zero.

Assume that there is an OBA ordering $\sigma'$ but \cref{algo:obaop} returns ``No ordering''. Let $X$ be the subset of $M$ containing all elements that has not been inserted into $\sigma$ during the execution of the algorithm. Let $x$ the leftmost element in $\sigma'$ that is an element of $X$. Since $x$ was not inserted into $S$ by the algorithm, the condition $t(x) > 0$ still holds at the end of \cref{algo:obaop}. This implies that there are sets $A$ and $B$ with $(A,B) \in {\cal R}$ and $A \subseteq X$ and $x \in B$. Thus, all elements of $A$ are to the right of $x$ in $\sigma'$. This is a contradiction to the fact that $\sigma'$ is an OBA ordering.

Let $\gamma = \sum_{A \in \mathcal{Q}} |A|$. To achieve the claimed running time we use the following data structure: Every set $A$ in $\cal Q$ is a linked list with pointers to its elements in $M$. Furthermore, $A$ has a linked list containing pointers to the tuples $(A,B) \in \R$. Every element $x \in M$ has a linked list containing the pointers to its position in the linked list of every set $A$ with $x \in A$. Obviously, all the steps before the \texttt{while}-loop can be done in time $\O(|M| + |\mathcal{R}| + \gamma)$. We iterate at most $|M|$ times through the \texttt{while}-loop. Every set in $\Q$ is visited at most once in the \texttt{foreach}-loop starting in line~\ref{line:obaop-for1}. The same holds for any tuple of $\R$ in the \texttt{foreach}-loop starting in line~\ref{line:obaop-for2}. As the $r$-value of each set $B$ becomes zero at most once, the set $B$ is visited at most once in the \texttt{foreach}-loop starting in line~\ref{line:obaop-for3}. Thus, the running time of the algorithm is bounded by $\O(|M| + |\mathcal{R}| + \gamma)$.
\end{proof}

    \section{Partial LBFS Orders of Chordal Bipartite Graphs}\label{sec:chordal-bipartite}

In~\cite{gorzny2017end-vertices}, Gorzny and Huang showed that the end-vertex problem of LBFS is \NP-complete on bipartite graphs but can be solved in polynomial time on AT-free bipartite graphs. In this section we will generalize the latter result in two ways by presenting a polynomial-time algorithm for the partial search order problem on chordal bipartite graphs, a superset of AT-free bipartite graphs. 

The following result will be a key ingredient of our approach. It shows that for two vertices $x$ and $y$ in the same layer $N_i(r)$ of a BFS starting in $r$ that have a common neighbor in the succeeding layer $N^{i+1}(r)$, it holds that the neighborhoods of $x$ and $y$ in the preceding layer $N^{i-1}(r)$ are comparable.

\begin{lemma}\label{lemma:neighborhoods}
Let $G$ be a connected chordal bipartite graph and let $r$ be a vertex of $G$. Let $x$ and $y$ be two vertices in $N^i(r)$. If there is a vertex $z \in N^{i+1}(r)$ which is adjacent to both $x$ and $y$, then $N(x) \cap N^{i-1}(r) \subseteq N(y)$ or $N(y) \cap N^{i-1}(r) \subseteq N(x)$.
\end{lemma}

\begin{proof}
Assume $x$ has a neighbor $u$ in $N^{i-1}(r)$ which is not a neighbor of $y$ and $y$ has a neighbor $v$ in $N^{i-1}(r)$ which is not a neighbor of $x$. Let $p_u$ and $p_v$ be shortest paths from $u$ to $r$ and $v$ to $r$, respectively, which shares the maximal number of edges. Let $w$ be the first vertex (starting from $u$ and $v$, respectively) that is both in $p_u$ and $p_v$. Let $p^*_u$ be the subpath of $p_u$ between $u$ and $w$ and let $p^*_v$ be the subpath of $p_v$ between $v$ and $w$. The paths $p^*_u$ and $p^*_v$ and the edges $xu$, $yv$, $xz$ and $yz$ form a cycle $C$ of length at least six in $G$. Due to the choice of $p_u$ and $p_v$ and the fact that two vertices of $G$ can only be adjacent if they lie in consecutive layers, the cycle $C$ is an induced cycle. This is a contradiction since $G$ is chordal bipartite.
\end{proof}

\Cref{algo:chordal-bipartite} presents the pseudocode of an algorithm for the rooted PSOP of LBFS on chordal bipartite graphs. We assume that the partial order $\pi$ contains only tuples where both elements are in the same layer of a BFS starting in $r$. Otherwise, the tuple is trivially fulfilled by any BFS ordering starting in $r$ or no such BFS ordering fulfills the tuple. The algorithm constructs an OBAP-instance with set $\Q_i \subseteq {\cal P}(N^i(r))$ and $\R_i \subseteq \Q_i \times \Q_i$ for any layer $i$ of the BFS. First we add the tuple $(\{x\},\{y\})$to the set $\R_i$ for every tuple $(x,y) \in \pi$ with $x,y \in N^i(r)$. Now the algorithm iterates through all layers starting in the last one. For any element $(A,B) \in \R_i$ the algorithm inserts a tuple $(A'',B')$ to the relation $\R_{i-1}$. The set $A''$ contains all neighbors of set $A$ in layer $i-1$ that are not neighbors of set $B$ and whose neighborhood in layer $i-2$ is maximal among all these neighbors. The set $B'$ contains all neighbors of $B$ in the layer $i-1$ that are not neighbors of $A$. At the end, the algorithm checks whether the OBAP-instance $(N^i(r), \Q_i, \R_i)$ of every layer $i$ can be solved. If this is not the case, then the algorithm rejects. Otherwise, it concatenates the computed OBA-orderings. The resulting ordering $\rho$ is used as tie-breaker for a LBFS$^+$ whose result is returned by the algorithm.

\begin{algorithm2e}[t]
\SetKw{KwDownto}{downto}
	\footnotesize
    \KwIn{Connected chordal bipartite graph $G$, vertex $r \in V(G)$, partial order $\pi$ on $V(G)$} 
    \KwOut{An LBFS ordering $\sigma$ of $G$ extending $\pi$ or ``$\pi$ cannot be linearized"}
    \Begin{
		let $k$ be the maximal distance of a vertex $v \in V(G)$ from $r$\;
		${\cal Q}_i \leftarrow \{\{x\}~|~x \in N^i(r)\} \qquad \forall i \in \{1,\ldots,k\}$\;
		$\R_i \leftarrow \{(\{x\},\{y\})~|~x,y \in N^i(r),~x \prec_\pi y\} \qquad \forall i \in \{1,\ldots,k\}$\;\label{line:cb1}
        \For{$i$ $\leftarrow$ $k$ \KwDownto 2}{
            \ForEach{$(A,B) \in \R_i$}{
                $A'$ $\leftarrow$ $[N(A) \cap N^{i-1}(r)]\setminus N(B)$\;
                $A''$ $\leftarrow$ $\{v \in A'~|~N(v) \cap N^{i-2}(r)=N(A') \cap N^{i-2}(r)\}$\;\label{line:cb3}
                $B'$ $\leftarrow$ $[N(B) \cap N^{i-1}(r)]\setminus N(A)$\;
                ${\cal Q}_{i-1}$ $\leftarrow$ ${\cal Q}_{i-1} \cup \{A'',B'\}$\;
                $\R_{i-1}$ $\leftarrow$ $\R_{i-1} \cup \{(A'',B')\}$\;\label{line:cb2}
            }
        }
        let $\rho$ be an empty vertex ordering\;
        \For{$i$ $\leftarrow$ $k$ \KwDownto 1}{
			\If{there is an OBA-ordering $\sigma$ for input ($N^i(r)$, ${\cal Q}_i$, $\R_i$)}{$\rho$ $\leftarrow$ $\sigma \mdoubleplus \rho$} 
			\lElse{\Return{``$\pi$ cannot be linearized"}}\label{line:cb4}
        }
        $\rho$ $\leftarrow$ $r \mdoubleplus \rho$\;
        \Return{LBFS$^+(\rho)$} of $G$\;
    }
    \caption{Rooted PSOP of LBFS on chordal bipartite graphs}\label{algo:chordal-bipartite}
\end{algorithm2e}

The following lemma is a direct consequence of the construction of the elements of $\R_i$ and \cref{lemma:neighborhoods}.

\begin{lemma}\label{lemma:algo-cb}
Let $(A,B) \in \R_i$. For any $x \in A$ it holds that $N(A) \cap N^{i-1}(r) \subseteq N(x)$ and if $B \neq \emptyset$ then there is a vertex $y \in B$ with $N(B) \cap N^{i-1}(r) \subseteq N(y)$.
\end{lemma}

\begin{proof}
The first claim follows immediately from the construction of $\R_i$ in \cref{algo:chordal-bipartite} (see lines~\ref{line:cb1} and~\ref{line:cb3}). 

For the second claim we first observe that the claim trivially holds if $(A,B)$ was inserted into $\R_i$ in line~\ref{line:cb1} as $|B| = 1$. Furthermore, it follows directly that the claim holds for all elements of $\R_k$. Now assume that the claim holds for all elements of $\R_i$ and let $(A',B')$ be an element of $\R_{i-1}$ with $B' \neq \emptyset$ that was inserted into $\R_{i-1}$ in line~\ref{line:cb2}. Then there is an element $(A,B) \in \R_i$ for which it holds that $B' \subseteq N(B) \cap N^{i-1}(r)$ and, by assumption on $\R_i$, there is a vertex $x \in B$ with $B' \subseteq N(x)$. This means that all elements of $B'$ have a common neighbor in $N^i(r)$. It follows from \cref{lemma:neighborhoods} that the neighborhoods of the elements of $B'$ in $N^{i-2}(r)$ can be ordered by inclusion. Thus, there must be a vertex $y \in B'$ for which it holds that $N(B') \cap N^{i-2}(r) = N(y) \cap N^{i-2}(r) \subseteq N(y)$.
\end{proof}

Using this lemma, we can show the correctness of \cref{algo:chordal-bipartite}.

\begin{theorem}\label{thm:algo:chordal-bipartite}
Given a connected chordal bipartite graph $G$, a partial order $\pi$ on $V(G)$ and a vertex $r \in V(G)$, \cref{algo:chordal-bipartite} decides in time $\O(|\pi| \cdot |V(G)|^2)$ whether there is an LBFS ordering of $G$ that starts in $r$ and is a linear extension of $\pi$.
\end{theorem}

\begin{proof}
We first assume that \cref{algo:chordal-bipartite} returns an ordering $\sigma$. Clearly, this ordering is an LBFS ordering. In the following we will show that the ordering $\sigma$ fulfills all OBA constraints given by the relations $\R_i$. Since all constraints of $\pi$ are covered by constraints of some $\R_i$ this implies that $\sigma$ is a linear extension of $\pi$. 

Assume for contradiction that $\sigma$ does not fulfill all constraints of the relations $\R_i$. Let $i$ be the minimal index such that there is a tuple $(A,B) \in \R_i$ that is not fulfilled by $\sigma$, i.e., there is a vertex $y \in B$ that is to the left of any element of $A$ in $\sigma$. Since $\rho$ fulfills the OBA constraints of $\R_i$, there is a vertex $x \in A$ that is to the left of $y$ in $\rho$. Due to \cref{lemma:plus}, there must be a vertex $z \prec_\sigma y$ with $zy \in E(G)$ and $zx \notin E(G)$. 
If $i = 1$, then this is not possible since the only neighbor of vertices in $N^1(r)$ that is to the left of some of them is the vertex $r$ and this vertex is adjacent to all vertices in $N^1(r)$. Thus, we may assume that $i > 1$. Then there is a tuple $(A'', B') \in \R_{i-1}$ for which it holds that $A'' \subseteq [N(A) \setminus N(B)] \cap N^{i-1}(r)$ and $B' = [N(B) \setminus N(A)] \cap N^{i-1}(r)$. The vertex $z$ is in $N(B)$. Due to \cref{lemma:algo-cb}, the neighborhood of $x$ in $N^{i-1}(r)$ is equal to the neighborhood of $A$ in $N^{i-1}(r)$ and, thus, $x$ is adjacent to all elements of $A''$. As $z \notin N(x)$, vertex $z$ is not in $N(A)$ and, thus, vertex $z$ is in $B'$. From the choice of $i$ it follows that the tuple $(A'',B')$ is fulfilled by $\sigma$, i.e., there is a vertex $v \in A''$ that is to the left of any element of $B'$ in $\sigma$. Vertex $v$ is not adjacent to $y$ but vertex $v$ is adjacent to $x$. Furthermore, any vertex $w \prec_\sigma v$ is either adjacent to $x$ or not adjacent to $y$ since otherwise $w$ would be an element of $B'$ to the left of $v$ in $\sigma$. Hence, the vertices $v$, $x$ and $y$ do not fulfill the 4-point-condition of LBFS in \cref{lemma:4point-lbfs}; a contradiction as $\sigma$ is an LBFS ordering of $G$.

Now assume that the LBFS ordering $\tau$ of $G$ is a linear extension of $\pi$. We will show that $\tau$ fulfills all OBA constraints in any set $\R_i$. Therefore, for any $\R_i$ the subordering of $\tau$ containing the elements of $N^i(r)$ is an OBA ordering fulfilling $\R_i$ and \cref{algo:chordal-bipartite} never reaches line~\ref{line:cb4}. Hence, it returns some LBFS ordering which is a linear extension of $\pi$ as was shown above.

Assume for contradiction that $\tau$ does not fulfill all OBA constraints. Let $i$ be the maximal index for which there is a tuple $(A'',B') \in \R_i$ that is not fulfilled by $\tau$, i.e., there is a vertex $y' \in B'$ that is to the left of any element of $A''$ in $\tau$. Since $\tau$ is a linear extension of $\pi$, the tuple $(A'',B')$ was inserted into $\R_i$ in line~\ref{line:cb2} and not in line~\ref{line:cb1}. Therefore, there is a tuple $(A,B) \in \R_{i+1}$ with $A'' \subseteq N(A) \cap N^i(r)$ and $B' = [N(B) \setminus N(A)] \cap N^i(r)$. Due to the choice of $i$, the constraint given by $(A,B)$ is fulfilled, i.e, there is a vertex $x \in A$ that is to the left of any element of $B$ in $\tau$.
 Let $y$ be an element of $B$ such that $N(B) \cap N^i(r) \subseteq N(y)$. As $B'$ is not empty, $B$ is not empty and the vertex $y$ must exist, due to \cref{lemma:algo-cb}. It holds that $yy' \in E(G)$ and $xy' \notin E(G)$. Due to \cref{lemma:4point-lbfs}, there is a vertex $z \prec_\tau y'$ with $zx \in E(G)$ and $zy \notin E(G)$. Let $z$ be the leftmost vertex in $\tau$ fulfilling this property. By the choice of $y$ and $y'$, the vertex $z$ is not an element of $N(B) \cap N^i(r)$ and not an element of $A''$. However, $z \in [N(A) \setminus N(B)] \cap N^{i}(r)$. Due to \cref{lemma:algo-cb}, there is a vertex in $A$ that is adjacent to all vertices in $N(A) \cap N^i(r)$. It follows from \cref{lemma:neighborhoods} that the neighborhoods in $N^{i-1}(r)$ of the vertices contained in $N(A) \cap N^i(r)$ are ordered by inclusion. Therefore, there must be a vertex $z'$ in the set $X = [N(A) \setminus N(B)] \cap N^{i}(r)$ whose neighborhood in $N^{i-1}(r)$ is equal to the neighborhood of the set $X$ in $N^{i-1}(r)$. Due to the construction of $A''$ in \cref{algo:chordal-bipartite}, the vertex $z'$ is in $A''$. Furthermore, it holds that $z \prec_\tau z'$, due to the choice of $z$. However, it holds that $N(z) \cap N^{i-1}(r) \subsetneq N(z') \cap N^{i-1}(r)$ as $z$ is not in $A''$. Note that any neighbor of $z$ to the left of $z$ in $\tau$ must be in $N^{i-1}(r)$ and, thus, it is also a neighbor of $z'$. Furthermore, there is at least one neighbor $a$ of $z'$ in $N^{i-1}(r)$ that is not adjacent to $z$. As vertex $a$ must be also to the left of $z$ in $\tau$, the vertices $a$, $z$ and $z'$ contradict the 4-point condition of LBFS given in \cref{lemma:4point-lbfs}. 
 
 It remains to show that the algorithm has a running time within the given bound. In a pre-processing step we compute the size of the neighborhood of every vertex in its preceding layer. This can be done in time linear in the size of the graph. Then we can compute the tuple $(A'', B') \in \R_{i-1}$ for any tuple $(A,B) \in \R_i$ in time $|V(G)|$ as the vertices in $A$ and $B$ with the most neighbors in the preceding layer are exactly the vertices that are adjacent to all neighbors of $A$ and $B$ in the preceding layer, due to \cref{lemma:algo-cb}. For any tuple in $\pi$, there is at most one tuple in every $\R_i$. Thus, the total number of tuples in the $\R$-sets is bounded by $|\pi| \cdot |V(G)|$ and the total size of the sets in these tuples is bounded by $|\pi| \cdot |V(G)|^2$. By \cref{thm:algo-obaop}, we need time $\O(|\pi| \cdot |V(G)|^2)$ in total to solve all OBA instances. The final LBFS$^+$ needs only linear time~(c.f. \cite{corneil2009lbfs}).
\end{proof}

Due to Observation~\ref{obs:end-vertex}, we can solve the end-vertex problem of LBFS on chordal bipartite graphs by solving the rooted PSOP $|V(G)|$ times with a partial order of size $\O(|V(G)|)$. This leads to the following time bound.

\begin{corollary}
Given a connected chordal bipartite graph $G$, we can solve the end-vertex problem of LBFS on $G$ in time $\O(|V(G)|^4)$ .
\end{corollary}

Similarly, it follows from Proposition~\ref{prop:f-trees} that the rooted \cf-tree recognition problem can be solved in time $\O(|V(G)|^4)$. Different to the general case, we can show that for BFS orderings of bipartite graphs the \cl-tree recognition problem can also be reduced to the partial search order problem. 

\begin{proposition}\label{prop:l-trees-bipartite}
The rooted \cl-tree recognition problem of any graph search $\cal A$ that produces BFS orderings can be solved on a bipartite graph $G$ by solving the rooted PSOP of $\cal A$ on $G$. 
\end{proposition}

\begin{proof}
For every vertex $y \neq r$ with parent $x$ in $T$, we insert the tuple $(x,y)$ to the partial order $\pi$. For any neighbor $z$ of $y$ that is in the same layer as $x$, we insert the tuple $(z,x)$ to $\pi$. We claim that any BFS ordering $\sigma$ extending $\pi$ has $T$ as its \cl-tree. Let $ab$ be an edge of $T$. W.l.o.g. vertex $a$ is the parent of $b$. From the construction of $\pi$ and $\sigma$ it follows that $a \prec_\sigma b$. Let $c$ be a neighbor of $b$ different from $a$ with $c \prec_\sigma b$. As $G$ is bipartite and $\sigma$ is a BFS ordering, vertex $c$ is in the same layer as the vertex $a$. Hence, it holds that $c \prec_\pi a$. As $\sigma$ is a linear extension of $\pi$, it also holds that $c \prec_\sigma a$. Thus, $c$ is not adjacent to $b$ in the $\cl$-tree of $\sigma$. This implies that $a$ is the parent of $b$ in this $\cl$-tree and, thus, the \cl-tree of $\sigma$ is equal to $T$.
\end{proof}

This proposition and the observation above lead to the following time bound for the search tree recognition problems on chordal bipartite graphs.

\begin{corollary}
On a chordal bipartite graph $G$, we can solve the rooted \cf-tree and the rooted \cl-tree recognition problem of LBFS in time $\O(|V(G)|^4)$.
\end{corollary}

    \section{Partial LBFS and MCS Orders of Split Graphs} 

Both the end-vertex problem and the \cf-tree recognition problem of several searches are well studied on split graphs (see~\cite{beisegel2019end-vertex,beisegel2021recognition,charbit2014influence}). In this section we will generalize some of these results to the partial search order problem. 

Consider a split graph $G$ with a split partition consisting of a clique $C$ and an independent set $I$. During a computation of an MNS ordering of $G$, every vertex that has labeled some vertex in $I$ has also labeled every unnumbered vertex contained in $C$. Therefore, we can choose a vertex of $C$ as the next vertex as long as there are still unnumbered vertices in $C$. This means that it is not a problem to force a clique vertex to be to the left of an independent vertex in an MNS ordering. However, forcing a vertex of $I$ to be to the left of a vertex of $C$ is more difficult. We will call a vertex of $I$ that is left to a vertex of $C$ in a vertex ordering $\sigma$ a \emph{premature vertex of $\sigma$}. The neighbors of such a premature vertex must fulfill a strong condition on their positions in $\sigma$ as the following lemma shows.

\begin{lemma}[{\cite[Lemma~22]{beisegel2018recognizing}}]\label{lemma:split-order}
Let $G$ be a split graph with a split partition consisting of the clique $C$ and the independent set $I$.  Let $\sigma$ be an MNS ordering of $G$. If the vertex $x \in I$ is a premature vertex of $\sigma$, then any vertex of $C$ that is to the left of $x$ in $\sigma$ is a neighbor of $x$ and any non-neighbor of $x$ that is to the right of $x$ in $\sigma$ is also to the right of any neighbor of $x$ in $\sigma$.
\end{lemma}

Similar to total orders we will call a vertex $x \in I$ a \emph{premature vertex of partial order $\pi$} if there is an element $y \in C$ with $x \prec_\pi y$. To decide whether a partial order $\pi$ can be extended by an MNS ordering the set of premature vertices of $\pi$ must fulfill strong properties which we define in the following. 

\begin{definition}
Let $G$ be a split graph with a split partition consisting of the clique $C$ and the independent set $I$. Let $\pi$ be a partial order on $V(G)$ and let $A$ be a subset of $I$. The tuple $(\pi, A)$ fulfills the \emph{nested property} if the following conditions hold:
\indent
\begin{enumerate}[leftmargin=1cm]
  \item[(N1)] If $y \in C$ and $x \prec_\pi y$, then $x \in C \cup A$.
  \item[(N2)] The neighborhoods of the elements of $A$ can be ordered by inclusion, i.e., there are pairwise disjoint sets $C_1, I_1, C_2, I_2, \ldots, C_k, I_k$ with $\bigcup_{j=1}^k I_j = A$ and for any $i \in \{1,\ldots,k\}$ and any $x \in I_i$ it holds that $N(x) = \bigcup_{j=1}^i C_j$.
  \item[(N3)] If $y \in C_i \cup I_i$ and $x \prec_\pi y$, then $x \in C_j \cup I_j$ with $j \leq i$.
  \item[(N4)] For any $i \in \{1,\ldots,k\}$ there is at most one vertex $x \in I_i$ for which there exists a vertex $y \in C_i$ with $x \prec_\pi y$.

\end{enumerate}
The \emph{nested partial order} $\pi^N(\pi, A)$ is defined as the reflexive and transitive closure of the relation containing the following tuples:
\begin{enumerate}[leftmargin=1cm]
  \item[(P1)] $(x,y) \quad \forall x,y \in V(G)$ with $x \prec_\pi y$
  \item[(P2)] $(x,y) \quad \forall x \in I_i \cup C_i, y \in V(G) \setminus \bigcup_{j=1}^i (I_j \cup C_j)$
  \item[(P3)] $(x,y) \quad \forall x \in C, y \in I \setminus A$ 
  \item[(P4)] $(x,y) \quad \forall x,y \in I_i$ for which $\exists z \in C_i \text{ with } x \prec_\pi z$
\end{enumerate}
\end{definition}

It is straightforward to check that $\pi^N(\pi, A)$ is a partial order if $(\pi, A)$ fulfills the nested property. We first show that the set $A$ of the premature vertices of a partial order $\pi$ must necessarily fulfill the nested property if there is an MNS ordering extending $\pi$. Furthermore, any such MNS ordering fulfills a large subset of the constraints given by the nested partial order $\pi^N(\pi, A)$.

\begin{lemma}\label{lemma:nested-necessary}
Let $G$ be a split graph with a split partition consisting of the clique $C$ and the independent set $I$ and let $\pi$ be a partial order on $V(G)$. Let $A = \{v \in I~|~\exists w \in C \text{ with }v \prec_\pi w\}$. If there is an MNS ordering $\sigma$ of $G$ extending $\pi$, then $(\pi, A)$ fulfills the nested property. If $x \prec_\sigma y$ but $(y,x) \in \pi^N(\pi, A)$, then $x \notin A \cup C$.
\end{lemma}

\begin{proof}
Due to the definition of $A$, the property (N1) trivially holds true. Let $\sigma$ be an MNS ordering extending $\pi$.

\emph{\textbf{Claim 1} If $x,y \in A$ and $x \prec_\sigma y$, then $N(x) \subseteq N(y)$:} As $y \in A$, there is a vertex $c \in C$ with $y \prec_\pi c$. Since $\sigma$ is a linear extension of $\pi$ it also holds that $x \prec_\sigma y \prec_\sigma c$. Due to \cref{lemma:split-order}, any neighbor of $x$ is to the left of $y$ in $\sigma$ and all these vertices are also neighbors of $y$. Hence, $N(x) \subseteq N(y)$. \hfill $\blacksquare$

Claim~1 implies that (N2) holds. Thus, we may assume in the following that the sets $C_i$ and $I_i$ defined in (N2) exist.

\emph{\textbf{Claim 2} If $y \in I_i$ and $x \prec_\sigma y$, then either $x \in C_j \cup I_j$ for some $j \leq i$ or $x \not\prec_\pi y$ and $x \notin A \cup C$: } 
As $y \in A$, there is a vertex $c \in C$ with $y \prec_\pi c$. If $x \in C$, then it follows from \cref{lemma:split-order} that $x \in N(y)$ and, therefore, $x \in C_j$ for some $j \leq i$. Thus, we may assume that $x \in I$. Due to Claim~1, it holds $N(x) \subseteq N(y)$. Thus, if $x \in A$, then $x \in I_j$ for some $j \leq i$. On the other hand, if $x \prec_\pi y$, then the transitivity of $\pi$ implies that $x \prec_\pi c$ and, thus, $x$ is an element of $A$. Thus, in any case Claim~2 holds. \hfill $\blacksquare$

Claim~2 implies that (N3) holds for all elements of all sets $I_i$ since $x \prec_\pi y$ implies that $x \prec_\sigma y$.

\emph{\textbf{Claim 3} If $y \in C_i$ and $x \prec_\sigma y$, then either $x \in C_j \cup I_j$ for some $j \leq i$ or $x \not\prec_\pi y$ and $x \notin A \cup C$:}
First assume that $x \in C$. If there is any vertex $z \in I_i$ with $x \prec_\sigma z$, then Claim~2 implies directly that $x$ fulfills the conditions of Claim~3. Thus, we may assume that there is a vertex $z \in I_i$ with $z \prec_\sigma x$. As $y \in N(z)$, it follows from \cref{lemma:split-order} that $x$ is also in $N(z)$ and, therefore, in some $C_j$ with $j \leq i$. Now assume that $x \in I$. It follows from \cref{lemma:split-order}, that $x \in N(y)$. Thus, if $x \in A$, then $x \in I_j$ for some $j \leq i$. On the other hand, if $x \prec_\pi y$, then $x$ is in $A$. Thus, in any case Claim~3 holds. \hfill $\blacksquare$  

Similar to Claim~2, Claim~3 implies that (N3) holds for any element of any $C_i$. Therefore, we have proven that (N3) holds in general.

To prove (N4), assume that there is vertex $x \in I_i$ and a vertex $y \in C_i$ with $x \prec_\sigma y$. Let $x$ be the leftmost vertex of $I_i$ in $\sigma$. By \cref{lemma:split-order} any other vertex $z \in I_i$ must be to the right of any neighbor of $x$. Thus, $z$ is to the right of any vertex $c \in C_i$ in $\sigma$ and $z \not \prec_\pi c$ for any $c \in C_i$. This implies property (N4).

Finally assume that there are vertices $x,y \in V(G)$ with $x \prec_\sigma y$ but $(y,x) \in \pi^N(\pi, A)$. The tuple $(y,x)$ cannot be a tuple of $\pi$ since $\sigma$ is a linear extension of $\pi$. If $y \in A \cup N(A)$, then it follows from Claims~2 and~3 that $x \notin A \cup C$. If $y \notin A \cup N(A)$ and $(y,x) \notin \pi$, then it follows from the construction of $\pi^N(\pi, A)$ that $x$ is an element of $I \setminus A$. Thus, in any case $x \notin A \cup C$.
\end{proof}

The nested property is, in a restricted way, also sufficient for the existence of a MNS ordering extending $\pi$. We show that if $(\pi, A)$ fulfills the nested property, then there is an MNS ordering that fulfills all tuples of $\pi$ that contain elements of the set $A$ or the clique $C$. This ordering can be found using an $\A^+$-algorithm. 

\begin{lemma}\label{lemma:nested-sufficient} 
Let $G$ be a split graph with a split partition consisting of the clique $C$ and the independent set $I$, let $\pi$ be a partial order on $V(G)$ and $A$ be a subset of $I$. Assume $(\pi, A)$ fulfills the nested property and let $\rho$ be a linear extension of $\pi' = \pi^N(\pi, A)$. Then for any graph search $\cal{A} \in \{$MNS, MCS, LBFS$\}$ the ordering $\sigma = \cal{A}^+(\rho)$ of $G$ fulfills the following property: If $x \prec_{\pi'} y$, then $x \prec_\sigma y$ or both $x$ and $y$ are not in $A \cup C$.
\end{lemma}

\begin{proof}
Assume for contradiction that $\sigma$ does not fulfill this property. Let $\pi' = \pi^N(\pi, A)$. Let $x$ be the leftmost vertex in $\sigma$ such that there is a vertex $y \in V(G)$ with $x \prec_\sigma y$, $y \prec_{\pi'} x$ and at least one of $x$ and $y$ are in $A \cup C$. Since $\rho$ is a linear extension of $\pi'$, it holds that $y \prec_\rho x$. Due to \cref{lemma:plus}, there must be a vertex $z \in V(G)$ with $ z \prec_\sigma x$, $xz \in E(G)$ and $yz \notin E(G)$. 

First assume that $y \in I$. If $x$ is in $A$ or in $C$, then $y \in A$, due to property (N3) or (N1), respectively. Otherwise, $y \in A$ by the choice of $x$ and $y$. Thus, in any case, we may assume that $y \in A$. Due to the choice of $x$, it must hold that $y \not \prec_{\pi'} z$. Let $I_i$ be the set containing $y$. Since $yz \notin E(G)$ and $y \not \prec_{\pi'} z$ it follows that $z \in \bigcup_{j=1}^i I_j$. Therefore, $x \in \bigcup_{j=1}^{i} C_j$ since $xz \in E(G)$. As $y \prec_{\pi'} x$ it follows from (N3) that $x \in C_i$ and, hence, $z \in I_i$. Furthermore, $(y,x)$ can only be an element of $\pi'$ if it is an element of $\pi$. However, this a contradiction to the choice of $x$ as now $y \prec_{\pi'} z$ holds true, due to (P4) in the definition of $\pi'$. 

So we may assume that $y \in C$. As $xz \in E(G)$ and $yz \notin E(G)$, it follows that $z \in I$ and $x \in C$. Since $N(z) \subseteq N(x) \cup \{x\}$ and $z \prec_{\sigma} x$ it follows from \cref{lemma:plus} that $z \prec_{\rho} x$ and, therefore, $x \not\prec_{\pi'} z$. By the definition of $\pi'$, it holds that $z \in A$. Let $z \in I_i$. Then $x \in C_i$ since $xz \in E(G)$ and $x \not\prec_{\pi'} z$. As $y \prec_{\pi'} x$, vertex $y$ must be in $C_j$ with $j \leq i$. Therefore, $yz \in E(G)$; a contradiction.
\end{proof}

After an instance of \cref{algo:ls} has visited all the clique vertices of a split graph, the labels of the remaining independent vertices do not change anymore. Thus, a vertex $x$ whose label is now smaller than the label of another vertex $y$ will be taken after $y$. Therefore, it is not enough to consider only the premature vertices of $\pi$. Instead, we must also consider all independent vertices $x$ that $\pi$ forces to be left of another independent vertex $y$ whose label is larger than the label of $x$ if all clique vertices are visited. In the case of MCS this is sufficient to characterize partial orders that are extendable. 

\begin{lemma}\label{lemma:split-mcs}
Let $G$ be a split graph with a split partition consisting of the clique $C$ and the independent set $I$. Let $\pi$ be a partial order on $V(G)$. Let $A := \{ u \in I~|~\exists v \in V(G) \text{ with } v \in C \text{ or } |N(u)| < |N(v)| \text{ such that } u \prec_\pi v\}$. There is an MCS ordering which is a linear extension of $\pi$ if and only if $(\pi, A)$ fulfills the nested property.
\end{lemma}

\begin{proof}
First we show that it is necessary that $A$ fulfills the nested property. To this end, we create the relation $\pi^*$ by adding the following tuples to $\pi$: If $u,v \in I$, $u \prec_\pi v$ and $|N(u)| < |N(v)|$, then $(u,w) \in \pi^*$ where $w$ is some element of $C \setminus N(u)$. Note that such a vertex $w$ must exist since $v$ has at least one neighbor that $u$ does not have. We claim that any MCS ordering $\sigma$ that extends $\pi$ is also a linear extension of $\pi^*$. 
Assume for contradiction that the tuple $(u,w)$ in $\pi^* \setminus \pi$ is not fulfilled by $\sigma$, i.e., $w \prec_\sigma u$. As $w \in C$ and $w \notin N(u)$, it follows from \cref{lemma:split-order} that all vertices of $C$ are to the left of $u$ in $\sigma$. Since $\sigma$ is a linear extension of $\pi$, it holds that $u \prec_\sigma v$. However, this is contradiction since the number of visited neighbors of $u$ would have been strictly smaller than the number of visited neighbors of $v$ when $u$ was chosen by the MCS. Therefore, $\sigma$ is a linear extension of $\pi^*$ and the reflexive and transitive closure $\pi^T$ of $\pi^*$ is a partial order. Note that a vertex $x$ is an element of $A$ if and only if there is a vertex $y \in C$ with $y \prec_{\pi^*} x$. Thus, by \cref{lemma:nested-necessary} it follows that $(\pi^T, A)$ fulfills the nested property. Since $\pi$ is a subset of $\pi^T$, the tuple $(\pi, A)$ also fulfills the nested property.

It remains to show that the property is also sufficient. Assume that $(\pi,A)$ fulfills the nested property. Let $\rho$ be a linear extension of $\pi^N(\pi, A)$ and $\sigma$ be the ordering produced by MCS$^+(\rho)$. Then, by \cref{lemma:nested-sufficient}, the ordering $\sigma$ fulfills the condition $x \prec_\pi y$ if $\{x,y\} \cap (A \cup C) \neq \emptyset$. It remains to show that $x \prec_\sigma y$ also holds for any pair $x,y \notin A \cup C$ with $x \prec_\pi y$. As $x$ is not in $A$, it must hold that $|N(x)| \geq |N(y)|$. Due to the definition of $\pi^N(\pi, A)$ and \cref{lemma:nested-sufficient}, for any vertex $z  \in C$ it holds that $z \prec_\sigma x$ and $z \prec_\sigma y$. Thus, vertex $x$ must be to the left of $y$ in $\sigma$ since its number of visited numbers was at least as large as the number of visited neighbors of $y$ and $x \prec_\rho y$.
\end{proof}

This lemma implies a linear-time algorithm for the PSOP of MCS on split graphs.

\begin{theorem}\label{thm:mcs-algo}
Given a split graph $G$ and a partial order $\pi$ on $V(G)$, we can solve the partial search order problem of MCS in time $\O(|V(G)| + |E(G)| + |\pi|)$.
\end{theorem}

\begin{proof}
First we compute the set $A$ by iterating through all elements $(u,v) \in \pi$ and checking whether one of the two properties given in \cref{lemma:split-mcs} hold. This can be done in time $\O(|\pi|)$. To check condition $(N2)$ in time $\O(|V(G)| + |E(G)|)$ we use the algorithm described in the proof of Proposition~14 in~\cite{beisegel2019end-vertex}. During this process, we can label every vertex in $A \cup N(A)$ with the index of its $I$-set or $C$-set and all other vertices with the label $\infty$. To check condition $(N3)$ we iterate through $\pi$ a second time. Now we check whether for all $x \prec_\pi y$ it holds that the label of $x$ is smaller or equal to the label of $y$ and whether there is at most one vertex $x$ in set $I_i$ for which there exists a vertex $y$ in $C_i$ with $x \prec_\pi y$. 
\end{proof}

This is a generalization of the linear-time algorithms for the end-vertex problem~\cite{beisegel2019end-vertex} and the \cf-tree recognition problem~\cite{beisegel2019recognizing} of MCS on split graphs. 

For LBFS there is a characterization of extendable partial orders that is similar to \cref{lemma:split-mcs}. However, due to the more complex label structure of LBFS, the result is slightly more complicated and uses OBA-orderings.

\begin{lemma}\label{lemma:split-lbfs}
Let $G$ be a split graph with a split partition consisting of the clique $C$ and the independent set $I$. Let $\pi$ be a partial order on $V(G)$. Let $A := \{ u \in I~|~\exists v \in V(G) \text{ with } v \in C \text{ or } N(u) \subsetneq N(v) \text{ such that } u \prec_\pi v\}$. Let $\pi'$ be the nested partial order $\pi^N(\pi, A)$ and let $\R$ be the following relation:

\vspace{-0.5cm}
\begin{align*}
{\cal R} = \{(X,Y)~|~&\exists x,y \in I \setminus A \text{ with } X = N(x) \setminus N(y),~Y = N(y) \setminus N(x),~x \prec_\pi y\} \\ &\cup \{(\{x\},\{y\})~|~x,y \in C, x \prec_{\pi'} y\}.
\end{align*}

There is an LBFS ordering extending $\pi$ if and only if the tuple $(\pi, A)$ fulfills the nested property and there is an OBA-ordering for $(C, {\cal Q}, {\cal R})$ where $\cal Q$ is the ground set of $\cal R$. 
\end{lemma}

\begin{proof}
Assume that the two conditions are fulfilled and let $\tau$ be an OBA ordering for the input $(C, {\cal Q}, {\cal R})$. Since $\tau$ fulfills the constraints of $\pi'$ on $C$, we can insert all elements of $I$ into $\tau$ such that the resulting ordering $\rho$ is a linear extension of $\pi'$. Let $\sigma$ be the LBFS$^+(\rho)$ ordering of $G$. Due to \cref{lemma:nested-sufficient}, for any $x,y \in A \cup C$ with $x \prec_\pi y$ it holds $x \prec_\sigma y$. Assume for contradiction that there are vertices $a,b \in I \setminus A$ with $a \prec_\pi b$ but $b \prec_\sigma a$. Since $\rho$ is a linear extension of $\pi$, it follows from \cref{lemma:plus} that $N(b) \setminus N(a)$ is not empty. 
Since $a \notin A$, it holds that $N(a) \not \subseteq N(b)$ and, therefore, the set $N(a) \setminus N(b)$ is also not empty. Due to the construction of $\R$ and $\rho$, there is a vertex $c \in N(a) \setminus N(b)$ that is to the left of any element of $N(b) \setminus N(a)$ in $\rho$. If this also holds in $\sigma$, then $b$ cannot be to the left of $a$ in $\sigma$, due to the 4-point condition of LBFS in \cref{lemma:4point-lbfs}. Thus, there must be a vertex $d \in N(b) \setminus N(a)$ that is to the left of $c$ in $\sigma$. Since $c \prec_\rho d$, it follows from  \cref{lemma:plus} that there is a vertex $z$ with $z \prec_\sigma d$, $zd \in E(G)$ and $zc \notin E(G)$. Since both $c$ and $d$ are elements of $C$, the vertex $z$ must be an element of $I$. If $z \notin A$, then $d \prec_{\pi'} z$. This contradicts \cref{lemma:nested-sufficient} since $z \prec_\sigma d$ and $d \in C$. Therefore, we may assume that $z \in A$. However, this means that $d$ is an element of some $C_i$ and $c$ is not an element of a $C_j$ with $j \leq i$. Thus, $d \prec_{\pi'} c$, a contradiction to the fact that $c \prec_\rho d$ and that $\rho$ is a linear extension of $\pi'$.

We will now show that the conditions are also necessary. Assume that the LBFS ordering $\sigma$ is a linear extension of $\pi$. First we show that $(\pi, A)$ fulfills the nested property in a similar way we have done it in \cref{lemma:split-mcs}. We create the relation $\pi^*$ by adding the following tuples to $\pi$: If $u,v \in I$, $u \prec_\pi v$ and $N(u) \subsetneq N(v)$, then $(u,w) \in \pi^*$ where $w$ is some vertex in $C \setminus N(u)$. Note that such a vertex $w$ must exist since $v$ has at least one neighbor that $u$ does not have. We claim that $\sigma$ is also a linear extension of $\pi^*$. Assume for contradiction that the tuple $(u,w)$ in $\pi^* \setminus \pi$ is not fulfilled by $\sigma$, i.e., $w \prec_\sigma u$. As $w \in C$ and $w \notin N(u)$, it follows from \cref{lemma:split-order} that all vertices of $C$ are to the left of $u$ in $\sigma$. Since $\sigma$ is a linear extension of $\pi$, it holds that $u \prec_\sigma v$. However, this is contradiction since the label of $u$ would have been strictly smaller than the label of $v$ when $u$ was chosen by LBFS. Therefore, $\sigma$ is a linear extension of $\pi^*$ and the reflexive and transitive closure $\pi^T$ of $\pi^*$ is a partial order. Note that a vertex $x$ is an element of $A$ if and only if there is a vertex $y \in C$ with $y \prec_{\pi^*} x$. Thus, by \cref{lemma:nested-necessary} it follows that $(\pi^T, A)$ fulfills the nested property. Since $\pi$ is a subset of $\pi^T$, the tuple $(\pi, A)$ also fulfills the nested property.

To show that the second condition also holds, we will prove that the ordering of the vertices of $C$ in $\sigma$ builds an OBA ordering for $(C, {\cal Q}, {\cal R})$. Assume for contradiction that there is a tuple $(X,Y) \in \R$ that is not fulfilled by $\sigma$, i.e., there is a vertex $b \in Y$ that is to the left of any vertex $a \in X$ in $\sigma$. For any $u,v \in C$ with $u \prec_{\pi'} v$ it follows from \cref{lemma:nested-necessary} that $u \prec_\sigma v$. Therefore, we may assume that $X = N(x) \setminus N(y)$ and $Y = N(y) \setminus N(x)$ for some $x,y \in I \setminus A$ with $x \prec_\pi y$. It holds that $x \prec_\sigma y$, that $by \in E(G)$ and that $bx \notin E(G)$. By the choice of $b$, for any vertex $z \prec_\sigma b$ the edge $yz$ is in $E(G)$ or the edge $xz$ is not in $E(G)$. If $b \prec_\sigma x$, then this implies that $\sigma$ is not an LBFS ordering, due to \cref{lemma:4point-lbfs}. Otherwise, we can observe that the set $X$ is not empty since in this case $N(x) \subsetneq N(y)$ would hold and this would imply that $x \in A$. Let $a \in X$, then it holds that $x \prec_\sigma b \prec_\sigma a$. This, however, contradicts \cref{lemma:split-order} as $x \in I$, $b \in C$, $bx \notin E(G)$ but $ax \in E(G)$.
\end{proof}

Again, this characterization leads to an efficient algorithm for the PSOP of LBFS on split graphs. However, its running time is not linear.

\begin{theorem}\label{thm:lbfs-algo}
Given a split graph $G$ and a partial order $\pi$ on $V(G)$, we can solve the partial search order problem of LBFS in time $\O(|V(G)| \cdot |\pi|)$.
\end{theorem}

\begin{proof}
First note that $|V(G)| \leq |\pi| \leq |V(G)|^2$ as $\pi$ contains all reflexive tuples. We can compute the set $A$ defined in \cref{lemma:split-lbfs} in time $\O(|\pi| \cdot |V(G)|)$ by iterating through the tuples of $\pi$ and comparing the neighborhoods of the two vertices in time $\O(|V(G)|)$. We check whether $A$ fulfills the property (N2) and compute the corresponding sets $I_i$ and $C_i$ in the same way we have done it in the proof of \cref{thm:mcs-algo}. Now we can compute the nested partial order $\pi' = \pi^N(\pi, A)$ in time $\O(|V(G)|^2)$. We extend the relation $\R$ defined in \cref{lemma:split-lbfs} by the set $S = \{(\{x\},\{y\})~|~x,y \in V(G),~x \prec_{\pi'} y\}$ and get the relation $\R'$. The size of $S$ is bounded by $\O(|V(G)|^2) \subseteq \O(V(G) \cdot |\pi|)$. The size of $\R' \setminus S$ is bounded by $\O(|\pi|)$ since we have inserted at most one tuple to $\R$ for any element of $\pi$. Since the size of the elements of the ground set of $\R'$ is bounded by $|V(G)|$, the total size of $\R'$, its ground set $\Q$ and the elements of $\Q$ is bounded by $\O(|V(G)| \cdot |\pi|)$. We compute an OBA ordering for the input $(V(G),\R',\Q)$ in time $\O(|V(G)| \cdot |\pi|)$ using \cref{algo:obaop}. Due to \cref{lemma:split-lbfs}, there is an LBFS ordering extending $\pi$ if and only if such an OBA ordering $\rho$ exists. The algorithm has a total running time in $\O(|V(G)| \cdot |\pi|)$.
\end{proof}

Unfortunately, the ideas of \cref{lemma:split-mcs,lemma:split-lbfs} cannot be directly adapted to the PSOP of MNS. A main difficulty of this problem seems to be the identification of independent vertices that have to be premature vertices. To illustrate this, we consider the example given in \cref{fig:split-example}. The defined partial order $\pi$ has no premature vertices. Furthermore, the set $A$ defined in \cref{lemma:split-lbfs} is empty for $\pi$. Nevertheless, for any MNS ordering $\sigma$ extending $\pi$, one of the vertices $f$ or $g$ has to be a premature vertex of $\sigma$.  

\begin{figure}[t]
\centering
\begin{tikzpicture}[vertex/.style={inner sep=2pt,draw,circle}, path/.style={decoration={snake, amplitude=0.3mm}, decorate}, edge/.style={-}, noedge/.style={dotted}, scale=1.5]
\footnotesize
\node[vertex,label={135:$a$}]  (1) at (0,0) {};
\node[vertex,label={45:$b$}]  (2) at (1,0) {};
\node[vertex,label={90:$c$}]  (3) at (0.5,0.75) {};
\node[vertex,label={135:$d$}]  (4) at (0,0.575) {};
\node[vertex,label={45:$e$}]  (5) at (1,0.575) {};
\node[vertex,label={180:$f$}]  (6) at (-0.5,0) {};
\node[vertex,label={0:$g$}]  (7) at (1.5,0) {};

\draw (1) -- (2) -- (3) -- (5) -- (2) -- (7);
\draw (1) -- (3) -- (4) -- (1) -- (6);
\end{tikzpicture}
\caption{A split graph consisting of clique $\{a,b,c\}$ and independent set $\{d,e,f,g\}$. Let $\pi$ be the reflexive and transitive closure of the relation $\{(f,e), (g,d)\}$. There is no MCS ordering extending $\pi$ since the set $A$ defined in \cref{lemma:split-mcs} contains both $f$ and $g$ and, thus, $(\pi, A)$ does not fulfill the nested property. There is neither an LBFS ordering extending $\pi$ as there is no OBA ordering for the relation $\R = \{(\{a\},\{b,c\}),~(\{b\}, \{a,c\})\}$ defined in \cref{lemma:split-lbfs}. However, the MNS ordering $(f,a,b,c,e,g,d)$ extends $\pi$.}\label{fig:split-example}
\end{figure}
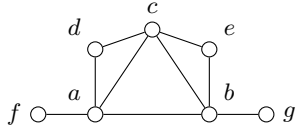

    \section{Further Research}
 
Besides the cases considered in this paper, there are several other combinations of graph classes and searches for which both the end-vertex problem and the \cf-tree recognition problem can be solved efficiently. Examples are the searches MNS and MCS on chordal graphs~\cite{beisegel2019end-vertex,beisegel2021recognition,cao2019graph}. Can all these results be generalized to the PSOP or is there a combination of graph search and graph class where the PSOP is hard but both the end-vertex problem and the \cf-tree recognition problem can be solved in polynomial time?

As mentioned in the introduction, the graph searches considered in this paper are used to solve several problems on graphs efficiently. This leads to the question whether the construction of a search ordering that extends a special partial order can be used in efficient algorithms for problems besides the end-vertex problem and the search tree recognition problem.

The algorithms given in this paper use the complete partial order as input. Using a Hasse diagram, it is possible to encode a partial order more efficiently. Since there are partial orders of quadratic size where the Hasse diagram has only linear size (e.g. total orders), it could be a good idea to study the running time of the algorithms for instances of the PSOP where the partial order is given as Hasse diagram.

	\bibliographystyle{plainurl}
	\bibliography{partial-orders}

\end{document}